\documentclass[letterpaper, 10 pt, conference]{ieeeconf}

\IEEEoverridecommandlockouts                              
\usepackage{cite}
\usepackage{amsmath,amssymb,amsfonts}

\usepackage{algorithm}
\usepackage{algpseudocode} 
\algtext*{EndWhile}
\algtext*{EndIf}

\usepackage{graphicx}
\usepackage{textcomp}
\usepackage{xcolor}

\usepackage[utf8]{inputenc}\usepackage{amsmath,amssymb,amsfonts}
\usepackage{graphicx}
\usepackage{textcomp}
\usepackage{comment}
\usepackage{multicol}
\usepackage{xcolor}
\usepackage{paralist} 
\usepackage{cite}
\usepackage{amsmath,amssymb,amsfonts}
\usepackage{graphicx}
\usepackage{textcomp}
\usepackage{xcolor}
\def\BibTeX{{\rm B\kern-.05em{\sc i\kern-.025em b}\kern-.08em
    T\kern-.1667em\lower.7ex\hbox{E}\kern-.125emX}}
    
\usepackage{nomencl}
\makenomenclature

\usepackage{tikz}
\usetikzlibrary{arrows,automata}
\usepackage[utf8]{inputenc}\usepackage{amsmath,amssymb,amsfonts}
\usepackage{todonotes}

\usetikzlibrary{shapes.geometric,fit}
\tikzstyle{container} = [draw, rectangle, inner sep=1.5cm]

  \tikzset{main node/.style={circle,draw,minimum size=1cm,inner sep=0pt},}
\usepackage{booktabs}
\usepackage{resizegather}

\usepackage{amsmath}
\usepackage{hyperref}
\usepackage{adjustbox}

\usepackage{amssymb}
\usepackage{acronym}
\usepackage{amsthm}
\usepackage{mathtools} 

\usepackage{tikz}
\usetikzlibrary{arrows,automata}
\usetikzlibrary{shapes.geometric,fit}
\tikzstyle{container} = [draw, rectangle, inner sep=1.5cm]

  \tikzset{main node/.style={circle,draw,minimum size=1cm,inner sep=0pt},}
\usepackage{algorithm}
\usepackage{algpseudocode} 
\usepackage{listings}

\newif\ifuseboldmathops
\newif\ifuseittextabbrevs
\useboldmathopstrue   

\ifuseittextabbrevs

	\newcommand{\ie}{{\it i.e.}}
	
	\newcommand{\etal}{{et~al.}}
\else

	\newcommand{\ie}{i.e.}
	
	\newcommand{\etal}{et~al.}
\fi

\ifuseboldmathops

\else

\fi

\ifuseboldmathops

\else

\fi

\ifuseboldmathops

\else

\fi

\ifuseboldmathops

\else

\fi

\newcommand{\Eventually}{\Diamond \, }

\newcommand{\until}{\mbox{$\, {\sf U}\,$}}

\newcommand{\dist}[1]{\mathcal{D}(#1)}

\newcommand{\calF}{\mathcal{F}}
\newcommand{\calAP}{\mathcal{AP}}

\newcommand{\win}{\mathsf{Win}} 

\acrodef{mdp}[MDP]{Markov Decision Process}
\acrodef{pomdp}[POMDP]{Partially Observable Markov Decision Process}

\theoremstyle{definition}
 \newtheorem{definition}{Definition}
 \newtheorem{example}{Example}
\newtheorem{problem}{Problem}

\newtheorem{lemma}{Lemma}
\newtheorem{assumption}{Assumption}
\newtheorem{proposition}{Proposition} \newtheorem{theorem}{Theorem}

\newcommand{\hgame}{\mathcal{HG}}

\newcommand{\pre}{\mathsf{Pre}}

\newcommand{\calA}{\mathcal{A}}
\newcommand{\game}{\mathcal{G}}

\acrodef{dfa}[DFA]{Deterministic Finite-State Automaton}
\acrodef{scltl}[scLTL]{syntactically co-safe LTL}
\acrodef{ltl}[LTL]{Linear Temporal Logic}

\renewcommand{\path}{\mathsf{Path}}

\title{Deceptive Labeling: Hypergames on Graphs for Stealthy Deception}
\author{Abhishek N. Kulkarni, Huan Luo, Nandi O. Leslie,  Charles A. Kamhoua, and Jie Fu
\thanks{A. Kulkarni and J. Fu are with the Robotics Engineering Program and Dept. of Electrical and Computer Engineering, Worcester Polytechnic Institute, Worcester, MA 01609 USA.
{\tt\small \{ankulkarni,jfu2\}@wpi.edu}}
\thanks{H. Luo is a visiting student with Dr. Fu at the WPI from Sept to Nov, 2019. {\tt\small hluo12@126.com}}
\thanks{N. Leslie and C. Kamhoua are  with U.S. Army Research Laboratory.
{\tt\small \{nandi.o.leslie.ctr, charles.a.kamhoua.civ\}@mail.mil}}
}

\begin{document}

\maketitle
\thispagestyle{empty}
\pagestyle{empty}

 \begin{abstract}
 With the increasing sophistication of attacks on cyber-physical systems, deception has emerged as an effective tool to improve system security and safety by obfuscating the attacker's perception. In this paper, we present a solution to the deceptive game in which a control agent is to satisfy a Boolean objective specified by a co-safe temporal logic formula in the presence of an adversary. The agent intentionally introduces asymmetric information to create payoff misperception, which manifests as the misperception of the labeling function in the game model. Thus, the adversary is unable to accurately determine which logical formula is satisfied by a given outcome of the game.
 We introduce a model called hypergame on graph  to capture the asymmetrical information with one-sided payoff misperception. Based on this model, we  present the solution of such a hypergame and use the solution to synthesize stealthy deceptive strategies. Specifically,  deceptive sure winning and deceptive almost-sure  winning strategies are developed by reducing the hypergame to a two-player game and one-player stochastic game with reachability objectives. A running example is introduced to demonstrate the game model and the solution concept used for strategy synthesis.
 \end{abstract}


\section{Introduction}
    \label{sec:introduction}
    
    With the increasing sophistication of the attacks on cyber-physical
    systems, deception has emerged as a tool to mitigate the strategic
    and informational disadvantages of the defender.   In this paper,
we consider a class of games where a control agent (P1, pronoun `he') plays against its adversary (P2, pronoun `she') to
satisfy a temporal logic formula, which describes high-level constraints such as safety, reachability, liveness, and reactivity \cite{Pnueli1989}. However, the task cannot be achieved if the
adversary knows the exact game.  
Thus, the agent needs to
falsify or obfuscate information to the adversary in order
to satisfy its temporal logic specification. The question arises, how to synthesize \emph{provably correct and deceptive} strategies that exploits the information advantages?

    The class of games where players' payoffs are Boolean valued (1 for satisfying the formula and 0 otherwise) is known as games on graphs (or $\omega$-regular games).
    The solution concepts of the games on graphs have been studied in formal synthesis of reactive systems \cite{gradel2002automata,bloemGraphGamesReactive2018a,chatterjee2012survey} and supervisory control \cite{lafortune2019discrete}. 
    However, existing work \cite{bloemGraphGamesReactive2018a,chatterjee2012survey} assumes both players have access to the correct model of the game. This is not the case when one player (deceiver, P1) can provide misleading information or intentionally hide information to the other for strategic advantages.
    In this paper, we study a class of asymmetric information games in which P1 has complete information about the game, but he intentionally \emph{falsifies or obfuscates P2's perception of one game component--the labeling function}, which maps an outcome (sequence of game states) to a Boolean payoff of \emph{one} if the temporal logic formula was satisfied or \emph{zero} otherwise. Such deception techniques are commonly used in decoy-based cybersecurity (such as honey-X) and defense (such as camouflage) \cite{jajodia2016cyber, handel2012war}.

To synthesize deceptive strategies for P1, we model the interaction between the two players as a hypergame \cite{bennettHypergameTheoryMethodology1986}. A hypergame models the situation where the players have different perceptions of their interaction given their information, and higher-order information (information about other's information).
We extend the normal-form hypergame model to define the
\emph{hypergame on graph} model to capture the perceptual games of the players and
their knowledge about the opponent's perceptual game. To solve for P1's
deceptive strategies, we adopt the solution concept of subjective
rationalizability \cite{sasakiHierarchicalHypergamesBayesian2016} from incomplete information game theory. A
subjective rationalizable player behaves rationally and assumes the
other player to act rationally in his/her subjective view of the game.
Thus, whenever P1 deviates from his rational strategy in P2's subjective view, we expect P2 to become aware of the information asymmetry.
Using this observation, we establish the
necessary and sufficient conditions for the deceptive strategies to be (a)
stealthy sure winning, and (b) stealthy almost-sure winning
(\ie,~winning with probability one). A stealthy strategy ensures
that P2 does not become aware of the information asymmetry until P1 can ensure to satisfy the temporal logic
specification irrespective of P2's actions. These solution concepts
for hypergames on graphs not only provide the provably-correct deceptive
strategies for P1 but also provide a way to assess
the effectiveness of deception and its potential limitations. 


\paragraph*{Related Work}
Game theory for deception has been investigated extensively using the two models of incomplete information games: hypergames \cite{bennettHypergameTheoryMethodology1986,wangSolutionConceptsHypergames1989,kovachHypergameTheoryModel2015} and Bayesian games \cite{carrollGameTheoreticInvestigation2009,al-shaerDynamicBayesianGames2019}. 
Hypergames were initially proposed and studied for the normal-form one-shot games \cite{bennettHypergameTheoryMethodology1986,wangSolutionConceptsHypergames1989} and later studied by Gharesifard and Cort\'es \cite{gharesifardStealthyDeceptionHypergames2014,gharesifardEvolutionPlayersMisperceptions2012} for repeated games. Gharesifard and Cort\'es developed an H-digraph model to monitor how a player's perception evolves during repeated interactions and to design stealthy deceptive strategy in which the deceiver's action does not contradict the perception of the mark. 
Bayesian games \cite{Harsanyi1967} are used to design deceptive
strategies in cybersecurity applications
\cite{carrollGameTheoreticInvestigation2009,al-shaerDynamicBayesianGames2019,thakoorCyberCamouflageGames2019a}. Dynamic
Bayesian games are used in \cite{al-shaerDynamicBayesianGames2019} for
active deception in cybernetwork, where the defender has incomplete
information about the type of the attacker (legitimate user or
adversary) and the attacker also is uncertain about the type of the
defender (high-security awareness or low-security awareness). 
Ornkar \etal \cite{thakoorCyberCamouflageGames2019a} formulate a
security game (Stackelberg game) to allocate limited decoy resources
in a cybernetwork to mask network configurations from the attacker.
Existing deception in games describes players' payoffs
 using rewards/costs. We choose to adopt the hypergame
model over Bayesian
games 
because the hypergame model facilitates the analysis of 
higher-order information. We also employ the subjective
rationalizability solution concepts in hypergames \cite{sasaki2014subjective}.

\section{Preliminaries}
\label{sec:prelim}

We begin with a brief overview of $\omega$-regular games
\cite{gradel2002automata}.  An $\omega$-regular game, hereafter referred
to as a game, is a tuple $\game = \langle G, \varphi \rangle$ which
consists of a game arena $G$, representing the dynamics of the
interaction between P1 and P2, and a \ac{ltl} specification $\varphi$
for P1. In this work, we consider turn-based, deterministic game
arenas and syntactically co-safe \ac{ltl} specifications. We formalize
these concepts below.



\paragraph*{\textbf{Game Arena}}

	A turn-based, deterministic game arena is a tuple 
$	    G = \langle S, A, T, s_0, \calAP, L \rangle
$
	where $S = S_1\cup S_2$ is a finite set of states partitioned into P1's states, $S_1$, and P2's states, $S_2$; $A = A_1 \cup A_2$ is the set of actions where $A_1$ and $A_2$ are the sets of actions for P1 and P2, respectively; $T : (S_1 \times A_1) \cup (S_2 \times A_2) \rightarrow S$ is a \emph{deterministic} transition function that maps a state-action pair to a next state. If there exists a state $s' \in S$ such that $T(s, a) = s'$, then we say that action $a$ is \emph{enabled} at $s$; $s_0 \in S$ is called the initial state of $G$; $\calAP$ is the set of atomic propositions; $L: S \rightarrow 2^\calAP$ is the labeling function that maps a state $s \in S$ to a subset $L(s)\subseteq \calAP$ of  propositions which evaluate `true' at $s$.

A \emph{path} $\rho=s_0s_1\ldots $ in $G$ is a sequence of states such that for any $i \ge 0$, there exists $a \in A$ for which $T(s_i, a) = s_{i+1}$. A path $\rho$ can be mapped to a word over $2^\calAP$ by using a labeling function $w = L(\rho) = L(s_0)L(s_1)\ldots $, which can be evaluated against logical formulas.  



\paragraph*{\textbf{Payoffs in Linear Temporal Logic}} 
    Given the set of atomic propositions, $\calAP$, an \ac{ltl} formula is inductively defined as:
    \[ 
        \varphi := \top \mid \bot \mid p \mid \varphi \mid \neg\varphi \mid \varphi \land \varphi \mid \bigcirc \varphi \mid \varphi {\until} \varphi, 
    \] 
    where  $\top,\bot$ are universally true and false, respectively,  $p \in \calAP$ is an atomic proposition,  $\bigcirc$  is a temporal operator called the ``next'' operator. $\until$ is a temporal operator called the ``until'' operator. 
    For details about the syntax and semantics of \ac{ltl},  readers are referred to \cite{Pnueli1989}.

In this work, we restrict the objectives of P1 to a subclass of \ac{ltl} called \ac{scltl} \cite{kupferman2001model}. An \ac{scltl} formula $\varphi$ is equivalently expressed as a \ac{dfa}, defined by a tuple
 $   \calA = \langle Q, \Sigma, \delta,\iota, F \rangle$
 which consists of a finite set $Q$ of states; a finite set of symbols $\Sigma = 2^\calAP$; a deterministic transition function $\delta: Q \times \Sigma \rightarrow Q$;  an unique initial state $\iota \in Q$; and a set $F$ of final states. We extend the transition function over words $u \in \Sigma^\omega$ to write $\delta(q,uw) = \delta( \delta(q, u), w)$. A word $w$ is \emph{accepted} by $\calA$ if and only if there exists a finite prefix $u$ such that $w = uv$ for some $v \in \Sigma^\omega$, and  $\delta(q, u)\in F$.
 Given a path $\rho$ in $G$, we say $\rho$ satisfies $\varphi$ over $G$, if and only if $L(\rho)$ is accepted by the \ac{dfa} $\calA$. 



 \paragraph*{\textbf{Zero-sum Game on a Graph}} 
 Given a game arena $G$ and the \ac{scltl} specification $\varphi$ of P1, 
	a zero-sum game on a graph is a tuple, $\game = (G, \varphi)$. 
	For a path   $\rho \in S^\omega$ in $G$, if the labeling  $L(\rho)$ satisfies $\varphi$, then the path is winning for P1.
	 Otherwise, it is winning for P2. 

	 Next, we construct a \emph{product game}  for solving the zero-sum game $\game$--that is,  determining from the initial state $s_0$, whether a player can enforce a  path winning for him, regardless of the actions of the other player.

    \begin{definition}[Product game]
        \label{def:product}
        Given an arena $G = \langle S, A, T, s_0, \calAP, L \rangle$ and a \ac{dfa} $\calA = \langle Q, \Sigma, \delta, \iota, F \rangle$ equivalent to the \ac{ltl} specification of P1 $\varphi$, the product game is a tuple
        $
        G \otimes \calA = \langle S \times Q, A, \Delta, (s_0, q_0), S \times F \rangle,
        $
    where 
        $S \times Q$ is a set of states partitioned into P1's states $S_1 \times Q$ and P2's states $S_2 \times Q$;
        $\Delta: (S_1 \times Q \times A_1) \cup (S_2 \times Q \times A_2) \rightarrow S \times Q$ is a \textit{deterministic} transition function that maps a game  state  $(s,q) $ and an action $a$ to a next state $(s',q')$ where $s' = T(s,a)$ and $q' = \delta(q, L(s'))$;  
        $(s_0, q_0) \in S \times Q$ where $q_0 = \delta(\iota, L(s_0))$ is the initial state of the product game;
        $S \times F \subseteq S \times Q$ is a set of final states. 
    \end{definition}

We slightly abuse the notation to denote the product game graph as $\game := G \otimes \calA$.
A path $\rho = (s_0, q_0), (s_1, q_1), \ldots$ in the product game $\game$ is a sequence of states in $\game$. By definition of this product game, the project of $\rho$ onto $S$, $s_0s_1\ldots $,  is a path in $G$ and satisfies $\varphi$
if and only if there exists $(s_i, q_i) \in \rho$ for some $i \geq 0$ such that $(s_i, q_i) \in S \times F$.  

Thus, P1 can win (or ensure a run to satisfy $\varphi$) by reaching the set $S\times F$ in the product game. P2 can win by always avoiding $S\times F$.  Thus, the product game is a \emph{reachability game} for P1 and a \emph{safety game} for P2. 


    
    In the product game, a randomized, memoryless\footnote{A memoryless strategy in $\game$ is a finite memory strategy in $G$, where the memory is represented by states in \ac{dfa} $\calA$.} \emph{strategy} for player $i$, for $i\in\{1,2\}$, is a function $\pi_i: S_i\times Q \rightarrow \dist{A_i}$, where $\dist{A_i}$ is the set of discrete probability distributions over $A_i$.
A   deterministic strategy $\pi_i:S\times Q\rightarrow A_i$ maps a state $(s,q)$ to an action.
We say that  player $i$ commits to a strategy $\pi_i$ if and only if for a given state  $(s,q)$, if $\pi_i(s,q)$ is defined, then an action is sampled from the distribution $\pi_i(s,q)$ (or the action $\pi_i(s,q)$ is taken if $\pi_i$ is deterministic), otherwise, player $i$ selects an action at random.   Let $\Pi_i$ be the set of strategies of player $i$. 
A strategy $\pi_1 \in \Pi_1$ is said to \emph{sure winning} at a state $(s, q) \in S \times Q$ if, for any $\pi_2 \in \Pi_2$, P1 is guaranteed to satisfy $\varphi$ within $0 \leq k < \infty$ steps for a determined upper bound  $k$  on the number of steps. 
A strategy $\pi_1 \in \Pi_1$ is said to \emph{almost-sure winning} at a state $(s, q) \in S \times Q$ if, for any $\pi_2 \in \Pi_2$, P1 is guaranteed to satisfy $\varphi$ with probability one, \ie,~P1 might require unbounded number of steps to satisfy $\varphi$.  
A pair $\langle \pi_1,\pi_2\rangle $ of strategies is a \emph{strategy profile}. 

The games in Def.~\ref{def:product} are determined
\cite{mcnaughton1993infinite,zielonka1998infinite}: From any state
$(s, q) \in S \times Q$ exactly one of P1 and P2 has a memoryless
sure winning strategy. 
This result allows us to partition the game
state space as $S \times Q = \win_1 \cup \win_2$. Here, $\win_1$
includes all the states from which P1 has a sure winning strategy and
$\win_2$ includes all the states from which P2 has a sure winning
strategy. 
Readers are referred to
\cite{zielonka1998infinite} and Chapter 2 of \cite{gradel2002automata}
for the details of the game solution.

\section{Game on Graph with Labeling Misperception}
\label{sec:main-result}

In security and defense applications, players often have incomplete asymmetric information about the game. For instance, in a decoy-based deceptive defense approach, only the defender knows which hosts are decoys but the attacker does not. Such situations can be understood as the attacker  `misperceives' the labels of  states in the  game. We  introduce a hypergame model to analyze the effect of P2 misperceiving the true labeling function and how P1 can leverage P2's misperception to synthesize deceptive strategies.


\subsection{Hypergame Model}


    
        
        

    \begin{definition}[Hypergame \cite{bennettHypergameTheoryMethodology1986}]
	A level-1 two-player  hypergame is a pair
	$\hgame^1 = \langle \game_1, \game_2 \rangle,$ where
	$\game_1,\game_2$ are games perceived by players P1 and P2,
	respectively.  A level-2 two-player hypergame is a pair
	$\hgame^2= \langle \hgame^1, \game_2 \rangle,$ where P1 perceives the
	interaction as a level-1 hypergame and P2 perceives the interaction
	as game $\game_2$. The first component of a hypergame is called \emph{perceptual game} of P1; and the second component is called \emph{perceptual game} of P2.
      \end{definition}
    
    While it is possible to define a level-$k$ hypergame (see \cite{wangSolutionConceptsHypergames1989}), we note that a level-2 hypergame is sufficient to model the game with asymmetric information, since P1 knows $\hgame^1$ and P2 is only aware of $\game_2$. 
    

\noindent \textbf{Information Structure} In this paper, we are interested in games with asymmetric information of labeling function. Specifically, both P1 and P2 know the following components $S$, $A$, $s_0$,  $T$ of the arena $G$, and P1's objective $\varphi$. However,  P1 has complete information about the labeling function, $L_1(s)= L(s)$ for all $s\in S$, and P2 has misperception: There exists at least one state $s\in S$,  $L_2(s) \neq L(s)$. Moreover,  
 P1 is aware of P2's perceived labeling function $L_2$. 
 
 This information structure captures  decoy-based deception and camouflage. For example, the attacker misperceives a honeypot to be a regular host and the defender is aware of the attacker's misperception.


\begin{definition}[Level-2 Hypergame with Labeling Misperception]
	\label{def:level2}
	Given the information structure and the perceptual games, $\game_1 = \langle G_1, \varphi \rangle$ with  $G_1 = \langle S, A, T, s_0, \calAP, L_1 \rangle$ and $\game_2 = \langle G_2, \varphi \rangle$ with $G_2 = \langle S, A, T, s_0, \calAP, L_2 \rangle$, the interaction between P1 and P2 is a level-2 hypergame $
	    \hgame^2 = \langle \hgame^1, \game_2 \rangle,
	$
	where $\hgame^1 = \langle \game_1, \game_2 \rangle$ is the level-$1$ hypergame.
\end{definition}

Given two games, $\game_1$ and $\game_2$, we can use their product
games to obtain the solutions, which yield different partitions of the
product state space $S\times Q$. Let $\win_1^k$
(resp,. $\win_2^k$) represent the winning region of P1 (resp., P2) in
$\game_k$. From the winning regions, the winning strategies can be
extracted by construction (See Chapter 2 of \cite{gradel2002automata} for
details). To illustrate the solution, we introduce a running example.



\begin{example}
\label{ex:1}
    In the game arena, $G$, (see Fig.~\ref{fig:arena}), we have two players: P1 (circle) and P2 (square). P1 chooses an action at a circle state, and P2 selects an action at a square state. Given the transitions are deterministic, we omit the action set and use the edges of the graph to refer to players' actions. 
    
    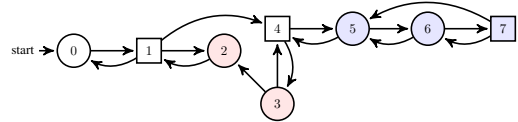
\begin{figure}[h]
        \centering
            \begin{tikzpicture}[->,>=stealth',shorten >=1pt,auto,node distance=2cm,
                            semithick, scale=0.5, transform shape,  square/.style={regular polygon,regular polygon sides=4}]
        \tikzstyle{every state}=[fill=white]
        \node[initial, state]   (0)                      {$0$};
        \node[square,draw]           (1) [right   of=0]   {$1$};
        \node[state, fill=red!10] (2) [right of =1] {$2$};
        \node[state, fill=red!10] (3) [below right of=2] {$3$};
        \node[square,draw] (4) [above of=3] {$4$};
        \node[state, fill=blue!10] (5) [ right of =4] {$5$};
        \node[state, fill=blue!10] (6) [right of=5] {$6$};
        \node[square, draw,, fill=blue!10] (7) [right of=6] {$7$};
        \path[->]   (0) edge node    {}   (1)
                    (1) edge [bend left] node {} (0)
                    (1) edge node {} (2)
                    (1) edge [bend left] node {} (4)
                    (2) edge [bend left] node {} (1)
                    (3) edge  node {} (2)
                    (3) edge  node {} (4)
                    (4) edge [bend left] node {} (3)
                     (4) edge  node {} (5)
                    (5) edge [bend left] node {} (4)
                     (5) edge  node {} (6)
                    (6) edge [bend left] node {} (5)
                     (6) edge  node {} (7)
                    (7) edge [bend left] node {} (6)
                      (7) edge  [bend right] node {} (5);
    \end{tikzpicture}
              \vspace{-2ex}
        \caption{A game arena, $G$. The red  (resp. blue) nodes are  P1's winning region  $\win_1^2$ in $\game_2$ (resp. $\win_1^1$ in $\game_1$).}
               \label{fig:arena}
    \end{figure}
    Let $L$ be defined such that $L(5) = \{A\}$ and
    $L(s) = \emptyset$ for $s\ne 5$. And $L_2 $ is defined such that  $L_2(2) = \{A\}$ and  $L_2(s) = \emptyset$ for $s\ne 2$. The objective of P1 is $\varphi = \Eventually A$,
    \ie,~eventually reaching a state labeled $A$. The \ac{dfa}
    equivalent to $\varphi$ is shown in Fig.~\ref{fig:dfa}.
    
    \begin{figure}[ht!]
        \centering
                \vspace{-2ex}
            \begin{tikzpicture}[->,>=stealth',shorten >=1pt,auto,node distance=2cm,
                            semithick, scale=0.5, transform shape,  square/.style={regular polygon,regular polygon sides=4}]
        \tikzstyle{every state}=[fill=white]
        \node[initial, state]   (0)                      {$0$};
        \node[accepting,state]           (1) [right   of=0]   {$1$};
        \path[->]   (0) edge node    {A}   (1)
                    (0) edge [loop above] node {$\emptyset$} (0)
                    (1) edge [loop right] node {$\emptyset$} (1);
    \end{tikzpicture}
              \vspace{-2ex}
        \caption{The \ac{dfa} for $\varphi= \Eventually A$.}
        \label{fig:dfa}
        \vspace{-2ex}
      \end{figure}
      Due to the simplicity of \ac{dfa}, we can directly solve
      $\game_1$ and $\game_2$ by marking the final set $\cal F$ for P1 to
      reach in the arena. In $\game_1$, the set to reach is
      $\{5\}$. The solution of $\game_1$ yields
      $\win_1^1 = \{5, 6, 7\}$ and $\win_2^1 = \{0, 1, 2, 3,
      4\}$. Whereas, In $\game_2$, the set to reach is $\{2\}$. The
      solution of $\game_2$ yields $\win_1^2 = \{2, 3\}$ and
      $\win_2^2 = \{0, 1, 4, 5, 6,
      7\}$. 
      In the true game $\game_1$, P2 can win the game by choosing the
      edge $(4, 3)$, but in her perceptual game $\game_2$, P2
      considers the action $(4, 5)$ to be winning or, in other words,
      rational.
\end{example}

We now introduce the solution concept of subjective rationality to
hypergames on graphs.
Let $S^+$ be paths of length $\ge 1$. Let $u_1: S^+  \times \Pi_1\times \Pi_2 \rightarrow [0, 1]$  be the utility function of P1 such that $u_1(\rho, \pi_1,\pi_2)$ is the probability of satisfying the specification $\varphi $ given that players commit to the strategy profile $\langle \pi_1,\pi_2 \rangle$ for a given history $\rho \in S^+$. 
The utility function for P2 is $u_2(\rho, \pi_1,\pi_2) =1- u_1(\rho,
\pi_1,\pi_2)$. We denote $u_i^j$ the \emph{utility function of player $i$ perceived by player $j$.}

\begin{definition}[Subjective Rationalizability]
Given a level-2 hypergame  $\hgame^2  = \langle \hgame^1 , \game_2
\rangle$ and the path $\rho \in S^+$,  strategy $\pi_i^\ast:
S^+\rightarrow \dist{A_i}$ (resp.,$\pi_j^\ast$) is \emph{subjective rationalizable} for P2 if and only if for all $\pi_i \in \Pi_i$, we have 
$u^2_i(\rho, \pi_i^\ast,\pi^\ast_{j}) \ge  u^2_i(\rho,  \pi_i,\pi^\ast_{j}),$
where $(i,j)\in \{(1,2), (2,1)\}$. The strategy $\pi_1^\ast$ is subjective rationalizable for P1 if and only for all $ \pi_1\in \Pi_1$,
$
u_1^1(\rho, \pi_1^\ast, \pi_2^\ast ) \ge 
u_1^1(\rho, \pi_1, \pi_2^\ast ),$
where $\pi_2^\ast$ is subjective rationalizable for player 2. 
\end{definition}

    
    


In words, a strategy is called subjectively rationalizable for  player $i$ if it is the best response in that player's perceptual game to some strategy of player $j$, which, for player $j$, is the best response to player $i$ in player $j$'s subjective view of player $i$'s perceptual game.  

  We now formally define the subjective rationalizable actions in $\game_2$. 


\begin{definition}[Subjective rationalizable actions in $\game_2$]
  For a given state $(s,q)$ in $\game_2$, an action of player $i$, for $i=1,2$, is subjective rationalizable for P2 if it has a non-zero probability to be selected by a subjectively rationalizable strategy of player $i$ in P2's perceptual game $\game_2$.
  \end{definition}

\begin{assumption}
    \label{assume:P2change-information-action}  
    Subjective rationalizability is a common knowledge between P1 and P2. 
\end{assumption}

Assumption~\ref{assume:P2change-information-action} means that both players know that their opponent is subjectively rational  and that the opponent is aware of this fact. Thus, we can say that P2 would become aware of her misperception, \ie,~$\game_2 \neq \game_1$, whenever P1 uses an action which is not subjectively rationalizable in P2's perceptual game, $\game_2$. 
Thus, we define the notion of a \textit{stealthy} deceptive winning strategy over a graphical model---a hypergame transition system---that effectively allows P1 to track histories in both $\game_1$ and $\game_2$.


\begin{definition}
    \label{def:hypergameTS}
    
    Given  games $\game_1 = \langle G_1, \varphi \rangle$ and $\game_2 = \langle G_2, \varphi \rangle$, a \emph{hypergame transition system} (HTS) is a tuple,
    \[
        \mbox{HTS} = \langle S \times Q \times Q, A, \Delta, (s_0, q_0, p_0), \win_1^1 \times Q \rangle,
    \]
    where 
    \begin{inparaenum}
        \item the transition function $\Delta$ is defined as follows:  given $(s, q, p), (s', q', p') \in S \times Q \times Q$, $\Delta((s, q, p), a) = (s', q', p')$ for some $a \in A$ if and only if $s' = T(s, a)$ and $q' = \delta(q, L_1(s'))$ and $p' = \delta(p, L_2(s'))$; and
        \item the initial state is $(s_0,q_0,p_0)$ where $s_0$ is the initial state in the game arena, $q_0 =  \delta(\iota, L_1(s_0))$, and $p_0 = \delta(\iota, L_2(s_0))$.
        \item $\win_1^1\times Q= \{(s,q,p)\mid (s,q)\in \win_1^1\}$.
    \end{inparaenum}
\end{definition}



     

\begin{definition}[Stealthy deceptive winning strategy]
A strategy $\pi_1: S\times Q\times Q \rightarrow \dist{A_1}$ defined on the $\mbox{HTS}$ is \emph{stealthy deceptive (sure/almost-sure) winning} in the hypergame $\hgame^2$ (in Def.~\ref{def:level2}) if the following two conditions are satisfied: 1)    \textsl{Stealthy}: For any $(s,q,p) \in S_1 \times Q \times Q\setminus \win_1^1\times Q$, then $\pi_1((s,q,p),a)>0$ only if action $a$
      is subjective rationalizable for P1 in $\game_2$;
     2) \textsl{Winning}: By committing to $\pi_1$, P1 ensures to visit a state in  $\win_1^1\times Q$, no matter which subjective rationalizable strategy that P2 commits to. 
    
A state $(s,q,p) \in S\times Q\times Q$ is \emph{stealthy deceptive (sure/almost-sure) winning} if P1 has a \emph{stealthy deceptive (sure/almost-sure) winning}  strategy at that state.
\end{definition}

We now formally state our problem:
\begin{problem}
Given a hypergame on graph $\hgame^2$ in Def.~\ref{def:level2} and
Assumption~\ref{assume:P2change-information-action}, how to synthesize a stealthy deceptive sure/almost-sure winning strategy for P1?
\end{problem}

       


\subsection{Synthesis of a stealthy deceptive sure winning strategy}
\label{sec:sure-win}
For P1's deceptive strategy to be stealthy, he must choose actions that are subjective rationalizable in P2's perceptual game until reaching the winning region $\win_1^1$. At the same time, a rational P2 takes subjective rationalizable actions in $\game_2$ unless she becomes aware of the misperception. 


\begin{lemma} In a turn-based deterministic perceptual game $\game_2$,
  for a state $(s,q)\in S\times Q$, an action $a$ is subjective
  rationalizable for player $i$ if and only if it satisfies either
  condition: 1) $(s,q)\in \win_i^2$ and $\Delta((s,q),a)\in \win_i^2$;
   2) $(s,q)\notin \win_i^2$ and $a$ is enabled from $s$. 
\end{lemma}
The first condition means that P2 thinks that a rational player should
stay within his/her winning region; the second condition means that P2
thinks that  it is rational for a player to take arbitrary actions if he/she
has already lost the game from that state.

 We introduce two functions $\pi_i^2$, for $i\in \{1,2\}$, that maps a state $(s,q) \in \win_i^2$ into a set of subjective rationalizable actions for player $i$ in the game $\game_2$. Formally, for each $i$, the function $\pi_i^2: \win_2^2 \cap( S_i \times Q) \rightarrow 2^{A_i}$ is defined such that, 
\begin{equation}
    \label{eq:sr_policy}
    \pi_i^2(s,q)  = \{a \mid \Delta_2((s,q), a) \in \win_i^2\}. 
\end{equation}

  \begin{theorem}
    \label{thm:sure-win}
    Given $\mbox{HTS} = \langle S\times Q\times Q, A,  \Delta, (s_0,q_0, p_0), \win_1^1 \times Q \rangle$, functions $\pi_2^2 : S\times Q\rightarrow 2^{A_2}$ and $\pi_1^2 : S\times Q\rightarrow 2^{A_1}$ defined by \eqref{eq:sr_policy}, P1 has a \emph{stealthy deceptive sure  winning} strategy if and only if he has a  \emph{sure winning} strategy in the following reachability game:
    \[
        \widetilde \hgame = (S\times Q\times Q, A, \tilde \Delta, (s_0,q_0,p_0), \win_1^1\times Q)
    \]
    where $\tilde \Delta$ is obtained from $\Delta$ by restricting both players' actions as follows: For a given state $(s,q,p)\in S\times Q\times Q$ and action $a\in A$,
if $(s,q)\in \win_1^1$, $\tilde \Delta((s,q,p),a) = \Delta((s,q,p),a)$, otherwise,

    \noindent \textsl{Case I}: $(s,p)\in \win_2^2  $ and $(s,q) \notin \win_1^1$, 
    
$   \tilde \Delta ((s,q,p),a) 
 =$\[\left\{  \begin{array}{ll}
 \Delta((s,q,p), a)  & \text{if $s\in S_1$}, \\
  \Delta((s,q,p),a) &   \text{if $s\in S_2$ and $a\in \pi_2^2(s,p)$},\\
    \uparrow    &  \text{if $s\in S_2$ and $a\notin \pi_2^2(s,p)$}.
\end{array}\right.
\]
where $\uparrow$ means that the transition is undefined.

\noindent \textsl{Case II}: $(s,p)\in \win_1^2$ and $(s,q) \notin \win_1^1$, 

$    \tilde \Delta ((s,q,p),a)= $
 \[  \left\{  \begin{array}{ll}
 \Delta((s,q,p), a)  & \text{if $s\in S_1$ and $a\in \pi_1^2(s,p)$ }, \\
     \uparrow    &  \text{if $s\in S_1$ and $a\notin \pi_1^2(s,p)$},\\
  \Delta((s,q,p),a) &   \text{if $s\in S_2$}.
\end{array}\right.
\] 
 The winning condition is defined by $\win_1^1\times Q$--that is, P1 wins if he reaches the set $\win_1^1 \times Q$.
\end{theorem}
\begin{proof}
Before reaching the set $\win_1^1\times Q$, at any state $(s,q,p)$ where $s\in S_2$, if $(s,p)$ is perceived winning by P2 (\ie,, $(s,p)\in \win_2^2$), then P2 will select a subjectively rationalizable action  $a\in \pi_2^2(s,p)$. If $(s,p)$ is not in $\win_2^2$, then any action from P2 is subjective rationzalizable.  At a state $(s,q,p)$ where $s\in S_1$, if $(s,p)\in \win_1^2$ but $(s,q)\notin \win_1^1$, then P1 will select a subjectively rationalizable action  $a\in \pi_1^2(s,p)$ so as not to contradict P2's perception.  If $(s,p)\notin \win_1^2$ and $(s,q)\notin \win_1^1$, then any action of P1 is deemed subjectively rationalizable by P2. The solution of reachability game $\widetilde \hgame$, is a policy $\pi_1^\ast: S\times Q\times Q\rightarrow 
A_1$ that ensures starting from a state where $\pi_1^\ast$ is defined, \emph{no matter which action P2} selects in $\tilde \hgame$, P1 can ensure to  reach a state $(s,q,p)$ with $(s,q)\in \win_1^1$ by following $\pi_1^\ast$, in finitely many steps. By construction, P2 will not know that a misperception exists as P1 takes only subjective rationalizable actions, until P1 reaches $\win_1^1$. After reaching the set, P1 can follow the true winning strategy defined for $\win_1^1$.   
\end{proof}

\begin{example}
 Given the \ac{dfa} shown in Fig.~\ref{fig:dfa},
we construct $\mbox{HTS}$ and $\widetilde \hgame$ shown in Fig.~\ref{fig:product-restrict}. In this figure, the red, dashed edges correspond to actions that are \emph{not}  subjective rationalizable in  P2's perceptual game and thus removed to obtain $\widetilde \hgame$. For example, at state $(3,0,0)$, P2 thinks that it is irrational for P1 to reach $(4,0,0)$ instead of $(2,0,1)$ given P2 misperceives the labels of states and thinks that P1 needs to reach state $2$.

\begin{figure}[H]
    \centering
            \vspace{-1ex}
        \begin{tikzpicture}[->,>=stealth',shorten >=1pt,auto,node distance=2.3cm,           semithick, scale=0.5, transform shape,  square/.style={rectangle}]

        \tikzstyle{every state}=[fill=white]
        \node[initial, ellipse,draw]   (000)                      {$0,0,0$};
        \node[square,draw]           (100) [right   of=000]   {$1,0,0 $};
        \node[ellipse,draw,dashed] (201) [right of =100] {$2,0,1 $};
        \node[square,draw,dashed] (101) [right of =201] {$1,0,1 $};
                \node[ellipse,draw,dashed] (001) [right of =101] {$0,0,1 $};
        \node[square, draw,dashed] (401) [below of=101] {$4,0,1$};
                \node[ellipse, draw,dashed] (301) [left  of=401] {$3,0,1$};
        \node[ellipse, draw,fill=blue!10,dashed] (511) [right of=401] {$5,1,1$};
                \node[square, draw,dashed] (411) [below of=511] {$4,1,1$};
        \node[ellipse, draw,fill=blue!10,dashed] (611) [right of=511] {$6,1,1$};
        \node[square, draw,fill=blue!10,dashed] (711) [below of=611] {$7,1,1$};
                       \node[ellipse, draw,dashed] (311) [below right of=411] {$3,1,1$};
        \node[square, draw] (400) [ below  of=000] {$4,0,0$};
        \node[ellipse,draw,dashed] (300) [right of=400] {$3,0,0$};
        \node[ellipse,draw,fill=blue!10] (510) [below of=400] {$5,1,0$};
        \node[square, draw] (410) [right of=510] {$4,1,0$};

        \node[ellipse,draw,dashed] (310) [right of=410] {$3,1,0$};
         \node[ellipse,draw,dashed] (211) [right of=310] {$2,1,1$};
        \node[ellipse,draw, fill=blue!10] (610) [below right of=510] {$6,1,0$};
                \node[square, draw,fill=blue!10] (710) [ right of=610] {$7,1,0$};
        \path[->]   (000) edge node    {}   (100)
                     (100) edge [bend left] node {} (000)
                    (100) edge [dashed, draw=red] node {} (201)         
                   (201) edge node {} (101)
              (101) edge [bend left, draw=blue!90, dash dot] node {} (201)

                (101) edge node {} (001)  
                (101) edge [draw=blue!90, dash dot]  node {} (401)
                (401) edge [draw=blue!90, dash dot]   node {} (511)
                (511) edge   node {} (611)
                (611) edge   node {} (711)
                (511) edge   node {} (411)
                (411) edge [bend left,draw=blue!90, dash dot]   node {} (511)
                (611) edge [bend left]  node {} (511)
                (411) edge [draw=blue!90, dash dot]   node {} (311)
                (311) edge [bend left]   node {} (411)
                (311) edge  node {} (211)
                (711) edge [bend left]  node {} (611)
                (401) edge  [draw=blue!90, dash dot] node {} (301)
                (301) edge  [ bend left]   node {} (401)

                (301) edge   node {} (201)
                (201) edge [bend left]    node {} (301)

                (001) edge [bend left] node {} (101)

                    
                    (100) edge   node {} (400)
                    (400) edge [dashed, draw=red] node {} (300)
                    (300) edge [dashed, draw=red, bend left]  node {} (400)
                    (300) edge node {} (201)
                   (400) edge  node {} (510)
                    (510) edge  node {} (410)
                    (410) edge [ bend left]  node {} (510)
                    (410) edge [dashed, draw=red] node {} (310)
                    (310) edge [dashed, draw=red, bend left] node {} (410)
                    (310) edge node {} (211)
                    (510) edge  node {} (610)
                    (610) edge [bend left]  node {} (510)
                    (610) edge  node {} (710)
                    (710) edge [bend left]  node {} (610);
    \end{tikzpicture}
        \vspace{-2ex}

    \caption{A graph representing $\mbox{HTS}$ and $\widetilde \hgame$. The blue and dash dot edges are deterministic choices in two-player reachability game $\widetilde \hgame$.  The red and dashed edges are not subjectively rationalizable for P2 and thus removed in $\widetilde \hgame$. Unreachable states in $\widetilde \hgame$ and $\hgame_M$ are drawn dashed.}
    \vspace{-2ex}
    \label{fig:product-restrict}
\end{figure}
 In the reachability game $\widetilde \hgame$, we  calculate the  stealthy deceptive sure winning region for P1, which includes $ \{(5,1,0), (6,1,0), (7,1,0), (4,1,0), (4,0,0)\}$. 
This means that P1 can satisfy his objective deceptively from states $\{4,5,6,7\}$--that is, one state more than the game where P2 does not have misperception. Due to P2's misperception, P2 will not select to go to state $3$ from state $4$--making the state $4$ deceptive sure winning for P1.

\end{example}
 

\subsection{Synthesis of a deceptive almost-sure winning strategy} 
\label{sec:almost-sure-deception}

In synthesizing the deceptive sure winning strategy for P1, we  assumed that P2 actively selects actions  in the  zero-sum game $\tilde \hgame$ to play against P1's objective. However, P2 cannot construct this hypergame transition system and thus may make ``mistakes'', exploitable by P1.
 To see this, let us consider the winning strategy for P2 in the reachability game $\widetilde \hgame$,  $\tilde \pi_2^{\ast}: S\times Q\times Q\rightarrow 2^{A_2}$.
  For P2 to exercise $\tilde \pi_2^\ast$, P2 should know the value of $q$ in the tuple $(s,q,p)$, which means that P2 should have a knowledge about $L_1$. This is not the case. Next, we consider a realistic assumption for P2.


\begin{assumption}
    \label{assume:randomP2}

    For a P2 state $(s,q,p)$ in the $\mbox{HTS}$, any  subjective rationalizable action at $(s,p)$ in $\game_2$ will be be selected by P2 with a non-zero probability. 
\end{assumption}


 The assumption on P2's behavior has the following rationale: At any given state,  the set of subjective rationalizable actions has the same values (either 1 or 0 depending on whether $(s,p)\in \win_2^2$). The assumption allows P2 to select any action in this set at random, instead of the worst-case scenario (considered by solving stealthy deceptive \emph{sure winning} strategy, in Sec.~\ref{sec:sure-win}). Besides, if P2 \emph{never} selects a subjective rationalizable action in her perceptual game, then the game is entirely different as we would have eliminated that action from  the arena. This P2's random choice of subjective rationalizable actions can be considered as opportunities for P1 to exploit.



\begin{theorem}
    Given $\mbox{HTS} = \langle S \times Q \times Q, A,  \Delta, (s_0,q_0, p_0),\win_1^1 \times Q \rangle$, P1 has a \emph{stealthy deceptive almost-sure  winning} strategy if and only if he has an \emph{almost-sure winning} strategy in the following one-player stochastic game:
    \[
        \hgame_M = (V=V_1\cup V_P, A_1, P,v_0, {\cal F}=\win^1_1\times Q ),
    \]
    where the  states are partitioned into two subsets: $V_1=S_1\times Q\times Q$ are a set of P1's states and $V_P= S_2\times Q\times Q$ are a set of \emph{probabilistic states}. 
    The transition function is partially defined as follows. First, any state in $\cal F$ is a sink or absorbing state.  At a  state $(s,q,p)\in V_1 \setminus {\cal F}$, we distinguish two cases:

    \noindent \textsl{Case I-1:} $(s,p)\in \win_2^2 $,  for any   action $a\in A_1$ enabled from $s$, $P ((s',q',p')| (s,q,p),a)=1$ where $(s',q',p')= \Delta((s,q,p),a)$.
 \textsl{Case I-2:}  $(s,p)\in \win_1^2$, for any   action $a\in \pi_1^2(s,p)$, $P((s',q',p')| (s,q,p),a)=1$ where $(s',q',p')=  \Delta((s,q,p),a)$.
 
    At a  state  $(s,q,p)\in V_P$, we distinguish two cases:
    \noindent \textsl{Case II-1: } $(s,p)\in \win_2^2$ , then for any  action $a\in \pi_2^2(s,p)$,     $P((s',q',p')| (s,q,p), a)>0$ where $(s',q',p')= \Delta((s,q,p),a)$.
    \textsl{Case II-2: } $(s,p)\in \win_1^2$, then for any   action $a\in A_2$ enabled from $s$, $P((s',q',p')| (s,q,p), a)>0$ where $(s',q',p')= \Delta((s,q,p),a)$.
    
 The almost-sure winning condition is defined by $\win_1^1\times Q$--that is, P1 wins if he reaches the set $\win_1^1 \times Q$ with probability one.
  \end{theorem}

The proof is similar to that of Thm.~\ref{thm:sure-win}, with small changes to consider randomized actions of P2. We omitted the proof due to the lack of space.

It is noted that only the support of $P((s,q,p),a)$ is known but not the exact probability distribution. 
The partial knowledge of the transition probability function gives us
a \emph{graph} of the underlying one-player stochastic
game. The  stealthy deceptive almost-sure
winning strategy for P1 is to ensure, with probability one, a state in
$\win_1^1\times Q$ can be reached.  
    Next, we describe Algorithm~\ref{alg:asw-mdp} to solve the almost-sure stealthy and deceptive winning strategy for P1.
   
\begin{algorithm}[t]
\caption{Computation of the almost-sure winning region and strategy for P1 in the one-player stochastic game.}
\label{alg:asw-mdp}
\begin{algorithmic}[1]
\item[\textbf{Inputs:}] $\hgame_M=(S= V_1\cup V_P, A_1, P,{ \cal F})$.
\item[\textbf{Outputs:} ] $X_k, \{Y_i\}$.
\State $X_0=V$, $Y_0={\cal F}$,  $k\leftarrow 0$,
\While{True}
 $i\leftarrow 0$,
\While{True}
\State $Y_{i+1} = \pre(Y_i,X_k)\cup Y_i $
\If{$Y_i=Y_{i+1}$}
\State Break.
\EndIf
\State  $i \leftarrow i
+1$.
\EndWhile
\If{$Y_i = X_k$} Break.
\EndIf
\State $X_{k+1}= Y_i$, $k\leftarrow k+1$, 
\EndWhile
 \end{algorithmic}
 \end{algorithm}
 The algorithm uses a function $\pre$ defined as follows.
 \begin{multline}
\pre(v, X) = \{ v'\in V_1 \mid \exists a \in A_1, P(v|v',a)=1\}\\
\cup 
\{v'\in V_P\mid P(v|v')>0 \implies v\in X\}
 \end{multline}
 and $\pre(Y,X)  = \cup_{v\in Y}\pre(v,X)$.

Intuitively, the set $\pre(Y,X)$ includes any state starting from which P1 can ensure to reach the set $Y$ with a positive probability, while staying in $X$ with probability one. The following result is readily obtained by construction. 

 \begin{proposition}
 The fix-point $X^\ast = X_k=X_{k+1}$ is the almost-sure winning region for P1 in the one-player stochastic game $\hgame_M$.
 \end{proposition}

Given the fixed point $X^\ast$, let $Y_0,Y_1,\ldots, Y_k$ be a sequence of states computed using $X= X^\ast$ in the inner loop, we can extract P1's deceptive almost-sure winning strategy $\pi_1$ as follows. For each $v\in Y_i \setminus Y_{i-1}$, $i> 0$,
 $
 \pi_1(v,a) =1 \text{ if } P(Y_{i-1} \mid v,a)=1
$.
 After reaching $\calF$, P1 follows his sure winning strategy in $\win^1_1$.
 


  \begin{example}

  The edges $(1,0,0)\rightarrow (0,0,0)$ and
  $(1,0,0)\rightarrow (4,0,0)$ are in Fig.~\ref{fig:product-restrict} are now
  probabilistic choices of
  P2. 
We compute 
\begin{inparaenum}
        \item $Y_0= \win^1_1\times Q = \{ (5,1,0), (6,1,0), (7,1,0)\}$. (here we omitted unreachable states.)
    \item $Y_1= \{(4,1,0), (4,0,0)\} \cup Y_0$.
    \item $Y_2= \{(1,0,0) \} \cup Y_1$,
    \item $Y_3= \{(0,0,0)\}\cup Y_2$.
\end{inparaenum}
Because $Y_4=Y_3$. The inner loop of Alg.~\ref{alg:asw-mdp} ends. Because now all reachable states in $X_0$ are in $Y_3$. We have $X_0 =Y_3$ and the outer loop of Alg.~\ref{alg:asw-mdp} ends.
Thus, the deceptive almost-sure  winning region  includes \emph{all states of the game}. 
 \end{example}
It is noted that   the solutions of deceptive strategies are based on solving multiple games (two-player zero-sum, turn-based games and one-player stochastic games). The space/time complexity  is linear in the size of \mbox{HTS} for solving the deceptive sure winning strategy, and polynomial for solving the deceptive almost-sure winning strategy.

\section{Conclusion and discussions}
\label{sec:conclude}
 
This paper presents a theory of hypergame for synthesizing stealthy  deceptive strategies with temporal logic specifications. We have shown that different from the games with complete information where the sure winning and almost-sure winning region overlap, the deceptive sure winning and almost-sure winning regions are different when one player has incomplete or incorrect information. 

This work lays the foundation for multiple future directions for both theoretical advances and algorithmic development.  One extension is to investigate the application of game-theoretic synthesis to cyber-physical security with decoy-based deception. This extension requires us to generalize the Assumption~\ref{assume:P2change-information-action} to incorporate other inference mechanisms. For example, if P2 can detect the true labeling after interacting with the decoy nodes, then P1 could include safety (prevent P2 from reaching decoys) as part of stealthy deception objective. In addition, we will extend the theory of hypergames to concurrent games on graphs \cite{DeAlfaro2007} and investigate the solution of this class of hypergames. 

\bibliographystyle{IEEEtran}
\bibliography{refs}

\end{document}